\documentclass {article}
\usepackage{txfonts}
\usepackage{bbding}
\usepackage{pifont}
\usepackage{mathbbol}

\usepackage[fleqn]{amsmath}
\newtheorem{myremark}{\bf  Remark}

\newtheorem{mydefinition}{\bf Definition}

\newtheorem{mytheorem}{\bf Theorem}

\usepackage{amsfonts}
\usepackage{tipa}
\usepackage{amsmath}
\usepackage{amssymb}
\usepackage{graphicx}
\usepackage{epsfig}
\usepackage{subfigure}
\usepackage{subfig}
\usepackage{indentfirst}
\usepackage{setspace}
\usepackage{caption}
\usepackage{colortbl,booktabs}
\usepackage{cite}
\definecolor{mygray}{gray}{.9}
\definecolor{mypink}{rgb}{.99,.91,.95}
\definecolor{mycyan}{cmyk}{.3,0,0,0}

\newcommand{\rmnum}[1]{\romannumeral #1}
\newcommand{\Rmnum}[1]{\expandafter\@slowromancap\romannumeral #1@}
\newenvironment{proof}[1][Proof]{\noindent\textbf{#1.} }{\ \rule{0.5em}{0.5em}}

\begin{document}

\baselineskip 0.7 cm
\title{ A fractional-order difference Cournot duopoly game\\ with long memory}
\author{\qquad Baogui Xin\thanks{B. Xin, College of Economics and Management, Shandong University of Science and Technology, Qingdao 266590, China, e-mail: xin@tju.edu.cn, corresponding author} \qquad\ Wei Peng\thanks{%
W. Peng, College of Economics and Management, Shandong University of Science and Technology, Qingdao 266590, China, e-mail: pengweisd@foxmail.com}\qquad\ Yekyung Kwon\thanks{%
Y. Kwon, Division of Global Business Administration, Dongseo University, Busan 47011, Korea, e-mail: yiqing@hanmail.net}}
\date{ }
\maketitle

\begin{abstract}
\baselineskip 0.7 cm We reconsider the Cournot duopoly problem in light of the theory for long memory. We introduce the Caputo fractional-order difference calculus to classical duopoly theory to propose a fractional-order discrete Cournot duopoly game model, which allows participants to make decisions while making full use of their historical information. Then we discuss Nash equilibria and local stability by using linear approximation. Finally, we detect the chaos of the model by employing a 0-1 test algorithm.

\noindent  {\bf Keywords:} Fractional-order difference; Fractional-order discrete dynamical systems; Cournot duopoly game; Bifurcation and chaos; 0-1 test.

\end{abstract}
\setcounter{tocdepth}{1}
\tableofcontents

\section{Introduction}
The purpose of this paper are twofold. One purpose is to bring together two independent lines of research in applied mathematics and industrial economics: fractional-order difference equations and Cournot equilibria. The other is to introduce long memory to classical Cournot duopoly games by replacing integer-order difference equations with fractional-order forms.

\subsection{Links between difference equation and Cournot equilibria in the literature}
Cournot game theory \cite{CournotA} was proposed in 1838. Many researchers, such as Nash \cite{Nash1951}, Von Neumann, and Morgenstern \cite{VonNeumann1944}, have since made landmark contributions to the development of game theory, which has become a powerful analytic tool in fields outside economics, such as cyberspace security \cite{HanL2019, SharehN, LalropuiaG, KhaliqS, RanjbarM }, power systems \cite{LiuH2019, WangQ }, cytobiology \cite{ArchettiK2019}, image processing \cite{Bhowmik}, human-machine systems \cite{LiangX}, artificial intelligence \cite{JiX2019, FarziY}, safety engineering \cite{GuanW, LiuQ }, nuclear security \cite{WooT}, oncology \cite{StankovaB}, system control \cite{WuY2019}, and information science \cite{MengS}. Considering the powerful dynamic characterization ability of difference equations, the discrete dynamical game with integer-order difference equations has become an important research direction in industrial economics field. Using integer-order difference equations, economics researchers have developed various interesting discrete Cournot duopoly game models, as shown in Table\ref{reftab1}.

\begin{table}[!h]
\centering \caption{Cournot duopoly games with integer-order difference equations}\label{reftab1}
\renewcommand\arraystretch{1.5}
\footnotesize\centering\begin{tabular}{ p{4.2cm}   p{1.2cm}<{\centering}  p{3cm}  p{5.2cm} }
\hline
Study & Number of players  &Decision variables &New characteristics \\
\hline
\rowcolor{mygray}Puu \cite{puuT1}  & 2 &product quantity & adaptive expectations\\
Kopel \cite{Kopel1996}, Govaerts \& Ghaziani \cite{Ghaziani2008}& 2 &product quantity &unimodal reaction functions\\
\rowcolor{mygray}Bischi \& Naimzada \cite{Bischi1999} & 2&product quantity &bounded rationality\\
Cavalli, Naimzada, \& Tramontana \cite{Cavalli2015}& 2 &product quantity&gradient-based approach, local monopolistic approximation\\
\rowcolor{mygray}Agliari, Naimzada, \& Pecora \cite{Agliari2016}& 2 &product quantity &differentiated products\\
Agiza \& Elsadany \cite{Agiza2003,Agiza2004}; Elsadany \& Awad \cite{Elsadany2016}; Fan et al. \cite{FanX2012} & 2 &product quantity & heterogeneous players\\
\rowcolor{mygray} Li, Xie, Lu, \& Tang \cite{LiX2016} & 2& callable product quantity & airline revenue management\\
El-Sayed, Elsadany, \& Awad \cite{ElsadanyE2015}& 2 &product quantity& logarithmic demand function\\
\rowcolor{mygray} Awad \& Elsadany \cite{ElsadanyA201601}& 2 &product quantity & social welfare\\
Elsadany \cite{Elsadany2017}& 2 &product quantity &bounded rationality based on relative profit maximization\\
\rowcolor{mygray}Matsumoto \cite{Matsumoto2006} & 2&product quantity & adaptive expectations\\
Fanti \& Gori \cite{Fanti2012}& 2 &product quantity & micro-founded differentiated products demand\\
\rowcolor{mygray}Xin, Ma, \& Gao \cite{XinMa2008} & 2&  product quantity & adnascent-type relationship\\
Ma \& Guo \cite{MaG2014}& 2 &product quantity& quantity estimation and two-stage consideration\\
\rowcolor{mygray}Askar, Alshamrani, \& Alnowibet \cite{AskarA2016}& 2 &product quantity & cooperation arising \\
Yu \cite{Yu2017}& 2&  car demand & heterogeneous business operation modes and differentiated productsn\\
\rowcolor{mygray}Zhang et al. \cite{ZhangZ2018}& 2 &product quantity&  two-stage, semi-collusion\\
Tramontana, Gardini, \& Puu \cite{Tramontana2009}& 2 &product quantity&competitors operate multiple production plants\\
\rowcolor{mygray}Tramontana \cite{Tramontana2010}& 2 &product quantity &isoelastic demand function\\
Tramontana et al. \cite{Tramontana2015} & 2-n &product quantity & heterogeneous players\\
\rowcolor{mygray}Baiardi \& Naimzada \cite{BaiardiN2018, BaiardiN201801, BaiardiN201802} & 2-n &product quantity & best response, imitation rules \\
\hline
\end{tabular}
\end{table}

\subsection{Links between fractional-order difference equations and scientific models in the literature}
D{\'{\i}}az and Osler \cite{Diaz1974} first presented the theory of fractional-order differences in 1974. Miller and Ross \cite{MillerR1} and Gray and Zhang \cite{Gray1988} developed fractional-order difference calculus in 1988. Fractional-order difference calculus has attracted the interest of many scholars and practitioners, but it is still a young research field compared to fractional differential calculus, which has developed for more than 300 years. Researchers of the former include Atici, Eloe, and {\c{S}}eng{\"{u}}l \cite{AticiF20070,AticiF2007,AticiE2009,AticiE200901,AticiS2010}; Bastos, Ferreira and Torres \cite{Torres2011}; Wu, Baleanu  et al. \cite{WuDH2018,WuDZ2018}; Goodrich \cite{Goodrich2016,Goodrich2018}; Abdeljawad et al. \cite{Abdeljawad2012,Abdeljawad2011,Abdeljawad2016,Abdeljawad2018}; Wei and Wang et al. \cite{Weiwang2011,Weiwang2018}; {\v{C}}erm{\'{a}}k, Gy{\H{o}}ri, and Nechv{\'{a}}tal \cite{Cermak2015}; Mozyrska and Wyrwas\cite{Mozyrska20171,Mozyrska20172}; and Abu-Saris and Al-Mdallal \cite{Abu-Saris2013}. Fractional-order difference calculus is applied in many scientific fields, as shown in Table \ref{reftab2}.

\begin{table}[!h]
\centering \caption{Non-mathematics applications of fractional-order difference calculus}\label{reftab2}
\renewcommand\arraystretch{1.5}
\footnotesize\centering\begin{tabular}{ p{3.2cm}   p{2.4cm}  p{3.2cm}  p{5cm} }
\hline
Study               & Research field     & Scientific problem  & Type of fractional-order difference calculus \\
\hline
\rowcolor{mygray} Tarasov \cite{TarasovV2015} &physics &  physical lattices & fractional-order difference equations\\
Wu et al. \cite{WuD2015}& physics & anomalous diffusion & fractional Riesz-Caputo difference equations\\
\rowcolor{mygray}Huang et al. \cite{HuangLL2016}; Ismail et al. \cite{IsmailS2017}; Liu, Xia, \& Wang \cite{LiuXW2018,LiuXW201801}; Kassim et al. \cite{Kassim2017} & cryptography & image encryption technique & fractional logistic difference equations\\
Xin et al. \cite {Xin2017}; Wu et al. \cite{WuDC2016}; Shukla \& Sharma \cite{ShuklaS2017}; Ouannas et al. \cite{OuannasA2018, OuannasA201801}; Liu \cite{LiuYC2016} & system engineering & chaos synchronization, control & fractional nonlinear difference equations\\
\rowcolor{mygray} Mozyrska  \& Pawluszewicz \cite{MozyrskaP2010,Mozyrska2014 ,MozyrskaW2015, MozyrskaP2015, Pawluszewicz2015} & system engineering & controllability, observability & fractional-order difference equations\\
Sierociuk \& Twardy \cite{Sierociuk2014}&system engineering & parameter identification & variable fractional-order difference equations\\
\rowcolor{mygray}Yin \& Zhou \cite{YinZ2015}& image processing & image denoising & difference curvature driven fractional nonlinear diffusion equations\\
Liu et al. \cite{LiuT2019}& signal processing & Kalman filter & nonlinear difference fractional system with stochastic perturbation\\
\hline
\end{tabular}
\end{table}

\subsection{Links between fractional-order difference equations and Cournot duopoly game}
The theory of integer-order difference equation is not suitable to analyze the nonlinear dynamic characteristics of the fractional discrete Cournot duopoly game with long memory. We must employ new theories to analyze the stability, bifurcation, and chaos of the fractional discrete game. Fortunately, researchers have produced some useful results for analyzing the stability of fractional-order difference equations, such as explicit stability conditions \cite{Cermak2015}, explicit criteria for stability \cite{Mozyrska20171}, stability by linear approximation \cite{Mozyrska20172, LiMa2013}, asymptotic stability criteria \cite{Abu-Saris2013}, Lyapunov functions \cite{WuLyapu2017}, finite-time stability \cite{WuFT2018}, and chaos analysis \cite{Khennaoui2019}.

The remainder of this paper is organized as follows. The relevant literature is reviewed in section 2. In section 3, a long-memory discrete Cournot duopoly game with Caputo fractional-order difference equations is presented. The Nash equilibrium points and their local stabilities are studied in section 4. In section 5, bifurcation diagrams, phase portraits, and 0-1 test algorithms are employed to validate the main results. This paper concludes with a summary in section 6.

\section{Preliminaries}
The following definitions of fractional-order difference calculus are introduced.

\begin{mydefinition} \label{def201}
(See\cite{AticiF20070}.) For any real numbers $\nu, t\in\mathbb{R} $, the $\nu$ rising fractional factorial of $t$ is defined as
\begin{equation*}
\begin{aligned}
t^{(\nu)} :=\frac{\Gamma(t+1)}{\Gamma(t+1-\nu)},\quad t^{(0)}=1.
\end{aligned}
\end{equation*}
\end{mydefinition}

\begin{mydefinition} \label{def202}
(See\cite{AticiF20070,Gray1988}.) Let $x:\mathbb{N}_a\rightarrow \mathbb{R}$, $a\in\mathbb{R}$, $t\in \mathbb{N}_{a+\nu}$, and $\nu>0$. Then the fractional sum of order $\nu$ is defined as
\begin{equation*}
\begin{aligned}
\Delta^{-\nu}_{a}x(t) :=\frac{1}{\Gamma(\nu)}\sum_{s=a}^{t-\nu}\left(t-\sigma(s)\right)^{(\nu-1)}x(s),\\
\end{aligned}
\end{equation*}
where $a$ is the start point, $\mathbb{N}_{a}=\{a,a+1,a+2,\ldots\}$ denotes the isolated time scale, and $\sigma(s)=s+1$.
\end{mydefinition}

\begin{mydefinition} \label{def203}
(See\cite{Abdeljawad2011}.) Let $\nu>0$, $\nu\notin\mathbb{N}$, $t\in \mathbb{N}_{a+n-\nu}$, and $n=[\nu]+1$. Then the $\nu$-order Caputo-like left delta difference is defined by
\begin{equation} \label{eq201}
\begin{aligned}
^{C}\Delta^{\nu}_{a}x(t) :=\Delta_{a}^{-(n-\nu)}\Delta^{n}x(t)=\frac{1}{\Gamma(n-\nu)}\sum_{s=a}^{t-(n-\nu)}\left(t-\sigma(s)\right)^{(n-\nu-1)}\Delta^{n}_{s} x(s).\\
\end{aligned}
\end{equation}
\end{mydefinition}

\begin{mytheorem} \label{theorem201}
(See\cite{Chenfl2011}.) For the Caputo fractional-order difference system
\begin{equation} \label{eq202}
\left\{\begin{aligned}
^{C}\Delta^{\nu}_{a}x(t) &=f\left(t^+, x\left(t^+\right)\right),\qquad t^+=t+\nu-1,\\
\Delta^{k} x(a)&=x_k, \qquad k=0,\ldots,m-1, \qquad m=[\nu]+1,\\
\end{aligned}\right.
\end{equation}
the equivalent discrete integral system is written as
\begin{equation} \label{eq203}
\begin{aligned}
x(t)=x_0(t)+\frac{1}{\Gamma(\nu)}\sum_{s=a+m-\nu}^{t-\nu}\left(t-\sigma(s)\right)^{(\nu-1)}f\left(s+\nu-1, x\left(s+\nu-1\right)\right),\qquad t\in \mathbb{N}_{a+m}\\
\end{aligned},
\end{equation}
where the initial iteration is
\begin{equation*}
\begin{aligned}
x_0(t)=\sum_{k=0}^{m-1}\frac{\left(t-a\right)^{(k)}}{\Gamma(k+1)}\Delta^{k}x(a).\\
\end{aligned}
\end{equation*}
\end{mytheorem}
Using Definition \ref{def201}, we can get
\begin{equation*}
\begin{aligned}
\left(t-\sigma(s)\right)^{(\nu-1)}=\left(t-s-1\right)^{(\nu-1)}=t^{(\nu)} :=\frac{\Gamma(t-s)}{\Gamma(t-s-\nu-1)}.
\end{aligned}
\end{equation*}
So, we can obtain the following proposition.
\begin{myremark}\label{rem201}
If\quad $a=0$, we rewrite system (\ref{eq202}) in the following numerical form
\begin{equation} \label{eq204}
\begin{aligned}
x(n)=x(0)+\frac{1}{\Gamma(\nu)}\sum_{i=1}^{n}\frac{\Gamma(n-i+\nu)}{\Gamma(n-i+1)}f\left( x(i-1)\right),\qquad n\in \mathbb{N}\\
\end{aligned}.
\end{equation}
\end{myremark}

According to the linearization theorem\cite{LiMa2013}, we obtain the following theorem, which is a special case of refs.\cite{Cermak2015,Mozyrska20172}.
\begin{mytheorem} \label{theorem202}
Assume that system (\ref{eq204}) is nonlinear with $v\in(0,1)$, $x(t)=(x_1(t),x_2(t),\cdots,x_n(t))^T$, and $f(t)=(f_1(t),f_2(t),\cdots,f_n(t))^T$ is continuously differentiable at a fixed point $x^{eq}$. Then $\forall t\in\mathbb{N}_{a+1-\nu}$, system (\ref{eq204}) is locally asymptotically stable when all eigenvalues $\lambda$ of the following Jacobian matrix
 \begin{equation*}
J(x^{eq})=\frac{\partial f(x)}{\partial x}\bigg|_{x=x^{eq}}=\left(
  \begin{array}{cccc}
    \frac{\partial f_{1}(x^{eq})}{\partial x_{1}} & \frac{\partial f_{1}(x^{eq})}{\partial x_{2}} & \cdots & \frac{\partial f_{1}(x^{eq})}{\partial x_{n}} \\
    \frac{\partial f_{2}(x^{eq})}{\partial x_{1}} & \frac{\partial f_{2}(x^{eq})}{\partial x_{2}} & \cdots & \frac{\partial f_{2}(x^{eq})}{\partial x_{n}} \\
    \vdots & \vdots &\ddots &\vdots \\
    \frac{\partial f_{n}(x^{eq})}{\partial x_{1}} & \frac{\partial f_{n}(x^{eq})}{\partial x_{2}} & \cdots & \frac{\partial f_{n}(x^{eq})}{\partial x_{n}} \\
  \end{array}\right)
\end{equation*}
are such that
\begin{equation*}
\begin{aligned}
\lambda\in\left\{z\in\mathbb{C}:|z|<\left(2\cos\frac{|\arg z|-\pi}{2-\nu}\right)^{\nu}\quad and\quad |\arg z|>\frac{\nu\pi}{2}\right\}.
\end{aligned}
\end{equation*}
\end{mytheorem}

Considering a two-dimensional case of system (\ref{eq204}), we obtain the following theorem, which is a special case of Theorem \ref{theorem202}.
\begin{mytheorem} \label{theorem203}
(See\cite{Cermak2015}.) Two-dimensional system (\ref{eq204}) is locally asymptotically stable if $\det J>0$ and either
\begin{equation}\label{eq205}
\begin{aligned}
\frac{-tr J}{2}\geq \sqrt{\det J},\quad \nu>\log_2\frac{\sqrt{\varpi}-tr J}{2},
\end{aligned}
\end{equation}
or
\begin{equation}\label{eq206}
\begin{aligned}
\frac{|tr J|}{2}< \sqrt{\det J}<\left(2\cos\frac{\kappa-\pi}{2-\nu}\right)^{\nu},\quad \nu<\frac{2\kappa}{\pi},
\end{aligned}
\end{equation}
where $\varpi=\big|(tr J)^2-4\det J\big|$, and $\kappa=\frac{tr J}{\sqrt{\varpi}}$.
\end{mytheorem}

In the following discussion, the time-dependence, subscripts, and superscripts will be omitted from notation if no confusion is caused.
\section{Model}

Let us consider a classical Cournot competition between firms 1 and 2, who produce homogeneous products which are perfect substitutes. Let $q_i(t), i=1,2$ denote the $i-$th firm's output during discrete time periods $t=0,1,2,\cdots$.  For simplicity, as mentioned by Bischi \& Naimzada \cite{Bischi1999} and  Agiza \& Elsadany \cite{Agiza2003,Agiza2004}, we assume the market-clearing price $p(t)$ at period $t$, an inverse demand function, is linear and decreasing, as follows:
\begin{equation}\label{ch301}
\begin{aligned}
p(t)=b-d(q_{i}(t)+q_{j}(t)),\qquad i,j=1,2,\quad i\neq j,
\end{aligned}
\end{equation}
where $b$ and $d$ are positive constants.

The  cost function takes the following form:
\begin{equation}\label{ch302}
\begin{aligned}
C_i(t)=\frac{1}{2}c_iq_i^2(t),\qquad i=1,2,
\end{aligned}
\end{equation}
where $c_i$, $i=1,2$, is a positive constant.

The $i-$th firm's profit can be written as
\begin{equation}\label{ch303}
\begin{aligned}
\Pi_i(q_{i}(t),q_{j}(t))=p(t)q_{i}(t)-C_i(t),\qquad i,j=1,2,\quad i\neq j.
\end{aligned}
\end{equation}

Then the $i-$th firm's marginal profit can be obtained by differentiating with respect to $q_i$:

\begin{equation}\label{ch304}
\begin{aligned}
\Phi_i(t)=\frac{\partial\Pi_i(q_{i}(t),q_{j}(t))}{\partial q_{i}(t)}=b-(c_i + 2d)q_{i}(t)-dq_j(t) ,\qquad i,j=1,2,\quad i\neq j.
\end{aligned}
\end{equation}

Now, let us consider the following repeated game mechanism. These two bounded rational firms have long memories of output decisions. Firms try to employ more complex adjustment mechanisms, which can be called ``bounded rationality with long memory." The dynamical adjustment mechanism is based on the long memory effect and the local estimation of the marginal profit, which can be described by
\begin{equation} \label{eq306}
\begin{aligned}
\Delta^{\nu}_{a}q_{i}(t) &=\alpha_{i}q_{i}(t^+)\Phi_i(t^+) ,\qquad i=1,2,
\end{aligned}
\end{equation}
where $\alpha_{i}>0$ is the speed of adjustment, and $^{C}\Delta^{\nu}_{a}$ is the $\nu$-order left Caputo-like delta difference. For the function $q(t)$, the  delta difference operator $\Delta$ is defined by $\Delta q(t)=q(t+1)-q(t)$.

\begin{myremark} \label{rem301}
When $\nu=1$ and \ $\Delta^{\nu}_{a}q_{i}(t)=q_{i}(t+1)-q_{i}(t)$, then equation (\ref{eq306}) degenerates to
\begin{equation} \label{eq307}
\begin{aligned}
q_{i}(t+1)=q_{i}(t)+\alpha_{i}q_{i}(t)\Phi_i(t) ,\qquad i=1,2,
\end{aligned}
\end{equation}
which agrees with the game model without long-memory effects proposed by Bischi \& Naimzada \cite{Bischi1999} and Agiza \& Elsadany \cite{Agiza2003,Agiza2004}.
\end{myremark}
From the above assumptions, we obtain the following discrete fractional Cournot duopoly game:
\begin{equation} \label{eq308}
\left\{\begin{aligned}
^{C}\Delta^{\nu}_{a}q_{1}(t) &=\alpha_{1}q_{1}(t^+)\left(b-(c_1 + 2d)q_{1}(t^+)-dq_2(t^+)\right) ,\\
^{C}\Delta^{\nu}_{a}q_{2}(t) &=\alpha_{2}q_{2}(t^+)\left(b-(c_2 + 2d)q_{2}(t^+)-dq_1(t^+)\right) .\\
\end{aligned}\right.
\end{equation}

\section{Nash equilibrium and local stability}

The equilibrium points of system (\ref{eq308}) satisfy
\begin{equation}\label{eq401}
\left\{ \begin{aligned}
&\alpha_{1}q_{1}\left(b-(c_1 + 2d)q_{1}-dq_2\right)=0,\\
&\alpha_{2}q_{2}\left(b-(c_2 + 2d)q_{2}-dq_1\right)=0.\\
\end{aligned} \right.
\end{equation}
By simple algebraic computation, we obtain four Nash equilibrium points:
$E_{1}=\left(0,0\right)$, $E_{2}=\left(0,\frac{b}{c_2+2d}\right)$, $E_{3}=\left(\frac{b}{c_1+2d},0\right)$, and
$E_{4}=(m(c_2+d),m(c_1+d))$, where $m=\frac{b}{c_1c_2+2c_1d+2c_2d+3d^2}$. They contain the following economic information:

 (\rmnum{1}) At the equilibrium point $E_{1}=(0,0)$, no single firm has anything to gain by producing its products if its opponent keeps producing nothing.

 (\rmnum{2}) At the equilibrium point $E_{2}=\left(0,\frac{b}{c_2+2d}\right)$, the best strategy of firm 1 is to produce nothing if firm 2 adopts its equilibrium production strategy $q_{2}^{*}=\frac{b}{c_2+2d}$. Similarly, the best output of firm 2 is $q_{2}^{*}=\frac{b}{c_2+2d}$ if firm 1 keeps producing nothing as its equilibrium strategy.

 (\rmnum{3}) At the equilibrium point $E_{3}=\left(\frac{b}{c_1+2d},0\right)$, firm 2 cannot obtain any payoff by producing its products if firm 1 adopts its equilibrium ouput strategy $q_{1}^{*}=\frac{b}{c_2+2d}$. In the same way, firm 1 will maintain its equilibrium output strategy $q_{1}^{*}=\frac{b}{c_2+2d}$ if firm 2 keeps zero output as its equilibrium strategy.

 (\rmnum{4}) At the equilibrium point $E_{4}=(m(c_2+d),m(c_1+d))$, firm 1 can receive no incremental benefit from deviating unilaterally from its chosen strategy $q_{1}=m(c_2+d)$, assuming firm 2 maintains its production strategy $q_2=m(c_1+d)$, and vice versa.

Obviously, equilibrium points $E_{1}$, $E_{2}$, and $E_{3}$ are bounded equilibria\cite{Bischi1999}. So, we only analyze the stability of non-bounded equilibrium point $E_{4}$, as follows.

\begin{mytheorem}\label{theorem401}
System (\ref{eq308}) is locally asymptotically stable at Nash equilibrium point $E_{4}$ if
\begin{equation}\label{eq402}
\begin{aligned}
\nu>\log_2\frac{\sqrt{(tr J)^2-4\det J}-tr J}{2}
\end{aligned},
\end{equation}
where
\begin{equation}\label{trJ}
\begin{aligned}
&tr J=-\frac{b(2d^2(\alpha_1+\alpha_2)+d(c_1(\alpha_1+2\alpha_2)+c_2(2\alpha_1+\alpha_2))+c_1c_2(\alpha_1+\alpha_2))}{3d^2+2d(c_1+c_2)+c_1c_2},
\end{aligned}
\end{equation}
\begin{equation}\label{detJ}
\begin{aligned}
&\det J=\frac{\alpha_1\alpha_2b^2(c_1+d)(c_2+d)}{c_1c_2+2d(c_1+c_2)+3d^2}.
\end{aligned}
\end{equation}
\end{mytheorem}

\begin{proof}
The Jacobian matrix $J$ of system (\ref{eq308}) evaluated at Nash equilibrium points $E_{4}$ is
\begin{equation*}
J(E_{4})=\left(
  \begin{array}{cc}
    \alpha_1(b-m(2c_1c_2+5d^2+d(3c_1+4c_2))) & -\alpha_1 md(c_2+d)\\
    -\alpha_2 md(c_1+d) & \alpha_2(b-m(2c_1c_2+5d^2+d(4c_1+3c_2)))
  \end{array}\right),
\end{equation*}
whose characteristic equation is
\begin{equation*}
p(\lambda)=\lambda^{2}-tr J\lambda+\det J,
\end{equation*}
where $tr J$ and $\det J$ are the same as (\ref{trJ}) and (\ref{detJ}), respectively.

We algebraically obtain
$\det J>0$, $tr J<0$, $(tr J)^2-4\det J>0$, and $\frac{-tr J}{2}\geq\sqrt{\det J}$. So, we can use Theorem \ref{theorem202} to find that system (\ref{eq308}) is locally asymptotically stable at Nash equilibrium point $E_{4}$ if eq. (\ref{eq402}) holds.

This completes the proof.
\end{proof}

\section{Numerical simulation}
According to Remark \ref{rem201}, system (\ref{eq501}) can be rewritten as
\begin{equation} \label{eq501}
\left\{\begin{aligned}
q_{1}(n) &=q_1(0)+\frac{1}{\Gamma(\nu)}\sum_{i=1}^{n}\frac{\Gamma(n-i+\nu)}{\Gamma(n-i+1)}\alpha_{1}q_{1}(i-1)\left(b-(c_1 + 2d)q_{1}(i-1)-dq_2(i-1)\right),\\
q_{2}(n) &=q_2(0)+\frac{1}{\Gamma(\nu)}\sum_{i=1}^{n}\frac{\Gamma(n-i+\nu)}{\Gamma(n-i+1)}\alpha_{2}q_{2}(i-1)\left(b-(c_2 + 2d)q_{2}(i-1)-dq_1(i-1)\right).\\
\end{aligned}\right.
\end{equation}
We will use eqs. (\ref{eq501}) to demonstrate stability, bifurcation, and chaos of system (\ref{eq308}).

Let parameters $\nu=0.99$, $\alpha1=0.45$, $\alpha2=0.12$, $b=6$, $d=4.1$, $c1=0.2$, and $c2=0.3$. Then non-bounded Nash equilibrium point $E_{4}=(0.4836, q2 = 0.4726)$, $tr J=-2.3101<0$, and $\det J=0.67378>0$. According to Theorem \ref{theorem401}, we can get $(tr J)^2-4\det J=2.6415>0$, $\frac{-tr J}{2}-\sqrt{\det J}=2.0079>0$, and $\nu=0.99>\log_2\frac{\sqrt{(tr J)^2-4\det J}-tr J}{2}=0.9765$, which is to say system (\ref{eq308}) is local asymptotically stable at non-bounded Nash equilibrium point $E_{4}$, as shown in Figure \ref{ffg501}.

\begin{figure}
\begin{center}
\begin{minipage}{140mm}
\subfigure[Firm 1.]{
\resizebox{7cm}{!}{\includegraphics{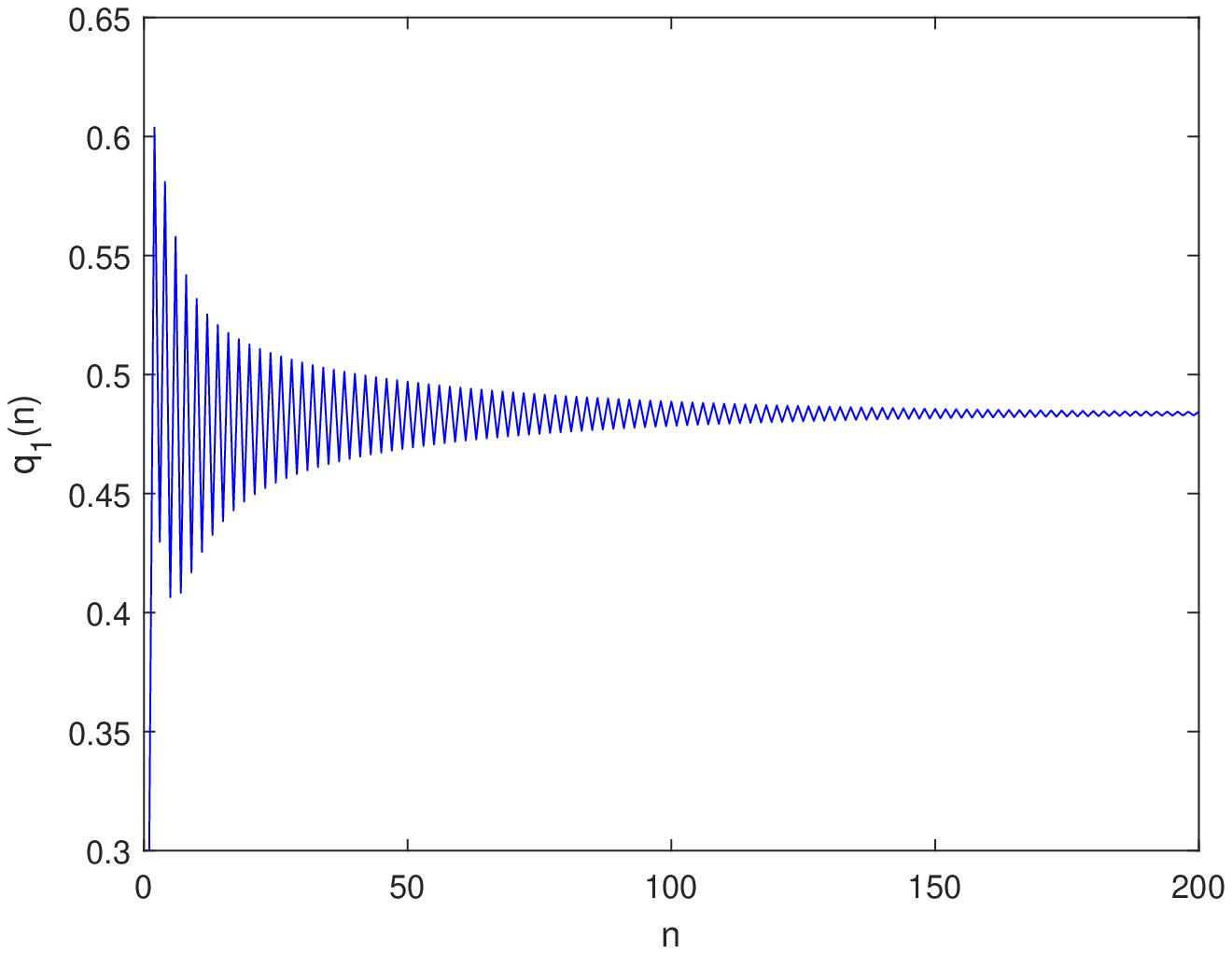}}}%
\subfigure[Firm 2.]{
\resizebox{7cm}{!}{\includegraphics{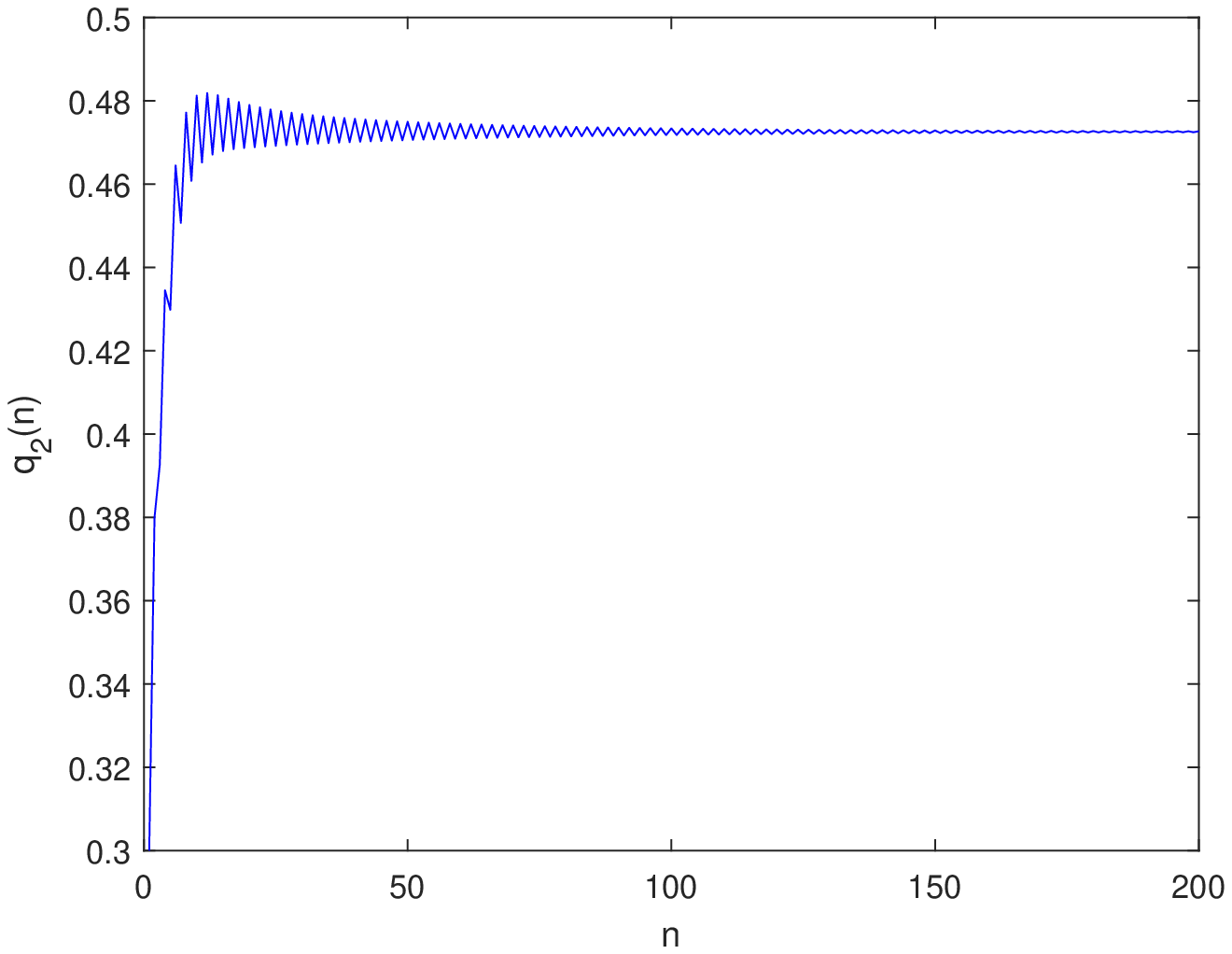}}}
\caption{Stabilized states of duopoly outputs of system (\ref{eq501}) with $(q_1(0),q_2(0))=(0.3,0.3)$.}
\label{ffg501}
\end{minipage}
\end{center}
\end{figure}

To analyze bifurcation and chaos of the long-memory system (\ref{eq308}) with fractional-order $\nu$, we vary the long-memory parameter, fractional-order $\nu\in(0,1)$, and fix the remaining parameters as $\alpha1=0.45$, $\alpha2=0.12$, $b=6$, $d=4.1$, $c1=0.2$, and $c2=0.3$, with the starting point $(q_1(0),q_2(0))=(0.3,0.3)$, as shown in Figure \ref{ffg502}. This figure includes the bifurcation diagram of firm 1's output and the scatter diagram of $K$, which is the  median value of the correlation coefficient of $q_1$. To obtain Figure \ref{ffg502}, the long-memory parameter $\nu$ varies from $0$ to $1$ with an increment of 0.002. The bifurcation diagram of $q_1$ is drawn by using 100 data points after dropping 500 transient data points. $K$ is calculated by taking 3,000 data points after dropping 500 transient data points. The bifurcation diagram demonstrates the possible long-term values of firm 1's output with long-memory parameter $\nu$ varying from $0$ to $1$. The scatter diagram of $\ K$ illustrates the possibility of chaos occurrence in system (\ref{eq308}) corresponding to different values of $\nu\in(0,1)$.

The criterion of the 0-1 chaos test\cite{Gottwald2009,Gottwald2008,Gottwald2007,BazineM2019,ZhuQi2018,RanLi2018,Yuan2019,Tiwari2019,Martinovic2018,Halfar2018,Vaidyanathan2018}
 shows that bounded trajectories in the $(p,s)-$plane or $K\approx0$ mean system (\ref{eq308}) is regular, and Brownian-like trajectories in the $(p,s)-$plane or $K\approx1$ indicate system (\ref{eq308}) is chaotic. So, system (\ref{eq308}) is chaotic when $\nu<0.4$, i.e., when $K\approx1$.  Obviously, system (\ref{eq308}) shows that the bifurcation diagram quite well coincides with the scatter diagram of $K$.

\begin{figure}
\centerline { \epsfig{figure=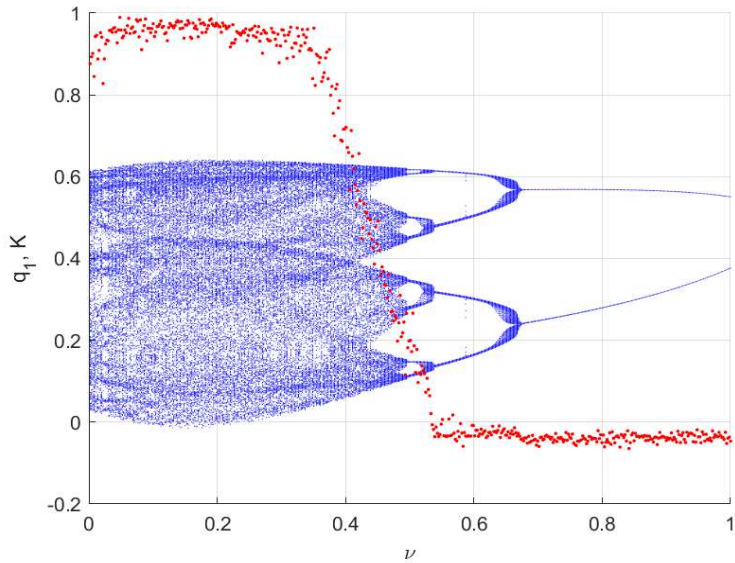, height=80mm, width=110mm}}
\caption{{Firm 1's output bifurcation (blue) and K (red) of system (\ref{eq501}) with varying $\nu\in(0,1)$.}}
\label{ffg502}
\end{figure}

For Figure \ref{ffg502}, let $\nu=0.2$, which leads to $K=0.9752$. We can know that system (\ref{eq308}) is chaotic, as shown in Figure \ref{ffg503}. To produce Figure \ref{ffg503}, we iterate it 3,500 times and draw a duopoly output time series using the first 300 data points, as shown in Figure \ref{ffg503}(a)-(b). We also create its phase portrait (as shown in Figure \ref{ffg503}(c)) and calculate its median correlation coefficient using the last 3,000 data points after dropping 500 transient data points, and draw its trajectories in new coordinates $(p, s)$, as shown in Figure \ref{ffg503}(d). From Figure \ref{ffg503}, we can find that their time series are irregular and chaotic, their output phase portrait is a strange chaotic attractor, and their trajectories in the $(p,s)-$plane are Brownian-like.
\begin{figure}
\begin{center}
\begin{minipage}{140mm}
\subfigure[Time series of firm 1's output.]{
\resizebox{7cm}{!}{\includegraphics{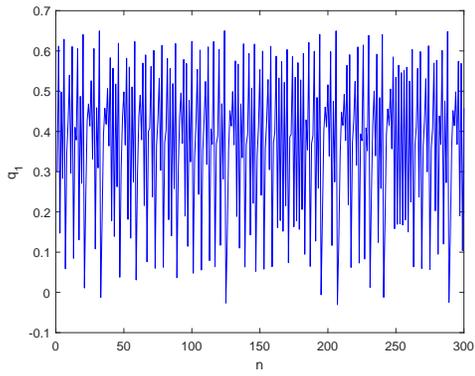}}}
\subfigure[Time series of firm 2's output.]{
\resizebox{7cm}{!}{\includegraphics{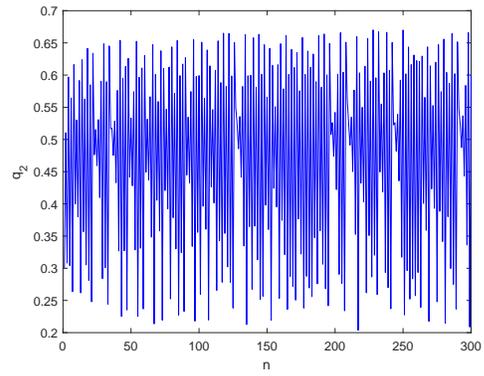}}}\\
\subfigure[Phase portrait of duopoly output.]{
\resizebox{7cm}{!}{\includegraphics{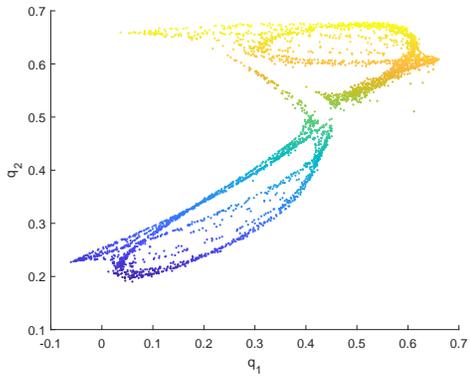}}}
\subfigure[Dynamics of the translation components $(p, s)$.]{
\resizebox{7cm}{!}{\includegraphics{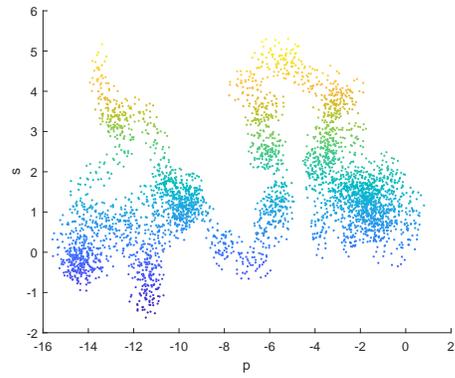}}}
\caption{Chaos in system (\ref{eq501}) with $(q_1(0),q_2(0))=(0.1,0.3)$.}
\label{ffg503}
\end{minipage}
\end{center}
\end{figure}

\section{Conclusion}

We have proposed a nonlinear fractional-order discrete Cournot duopoly game model that shows the long-memory effect. We discussed its Nash equilibria and local stability by linear approximation, then numerically illustrated its phase portraits, bifurcation diagrams, and chaos attractors using the 0-1 test algorithm. The paradigm to analyze the Cournot duopoly game can be applied to other scientific fields. There are still some open issues regarding fractional-order difference calculus, such as qualitative analysis theories of bifurcation, and high efficient Lyapunov algorithms.

\section*{Acknowledgments}
This work is supported by Natural Science Foundation of Shandong Province ( Grant No. ZR2016FM26) and  National Social Science Foundation of China ( Grant No.16FJY008).

\end{document}